\documentclass{siamart1116}

\usepackage{color,verbatim}
\usepackage{amsfonts,mathtools,bm}
\usepackage{graphicx,subcaption}

\newsiamremark{example}{Example}
\newsiamremark{remark}{Remark}

\def\eqns#1{\begin{equation*}#1\end{equation*}}
\def\eqnl#1#2{\begin{equation}\label{#1}#2\end{equation}}
\def\eqnsa#1{\begin{subequations}\begin{align*}#1\end{align*}\end{subequations}}
\def\eqnla#1#2{\begin{subequations}\label{#1}\begin{align}#2\end{align}\end{subequations}}

\def\zero{\mathbf{0}}

\def\bse{\bm{e}}
\def\bsk{\bm{k}}
\def\bsn{\bm{n}}
\def\bsN{\bm{N}}
\def\bsm{\bm{m}}
\def\bsr{\bm{r}}
\def\bsx{\bm{x}}
\def\bsalpha{\bm{\alpha}}
\def\bsbeta{\bm{\beta}}

\def\calF{\mathcal{F}}
\def\calI{\mathcal{I}}
\def\calN{\mathcal{N}}

\def\bbE{\mathbb{E}}
\def\bbI{\mathbb{I}}
\def\bbN{\mathbb{N}}
\def\bbP{\mathbb{P}}

\def\d{\mathrm{d}}

\DeclareMathOperator{\supp}{supp}
\DeclareMathOperator{\var}{\textbf{var}}

\def\AND{\qquad\mbox{ and }\qquad}
\def\vphi{\varphi}

\def\defeq{\doteq}
\def\given{\,|\,}
\def\supn#1{|#1|_{\infty}}
\def\st{\,:\,}
\def\dotleq{\mathbin{\dot\leq}}
\def\dotle{\mathbin{\dot<}}

\def\est#1{#1}

\title{Unbiased Multi-index Monte Carlo}

\author{Dan Crisan%
\thanks{Department of Mathematics, Imperial College London, London, SW7 2AZ, UK.}
\and
Pierre Del Moral%
\thanks{INRIA Bordeaux, 200 Avenue de la Vieille Tour, 33405 Talence, FR.}
\and
Jeremie Houssineau%
\thanks{Department of Statistics and Applied Probability, National University of Singapore. Email: \href{mailto:stahje@nus.edu.sg}{stahje@nus.edu.sg}}
\and
Ajay Jasra%
\thanks{Department of Statistics and Applied Probability, National University of Singapore.}
}

\headers{Unbiased Multi-index Monte Carlo}{D.\ Crisan, P.\ Del Moral, J.\ Houssineau and A.\ Jasra}

\begin{document}

\maketitle

\begin{abstract}
We introduce a new class of Monte Carlo based approximations of expectations of random variables such that their laws are only available via certain discretizations. Sampling from the discretized  versions of these laws can typically introduce a bias. In this paper, we show how to remove that bias, by introducing a new version of multi-index Monte Carlo (MIMC) that has the added advantage of reducing the computational effort, relative to i.i.d.~sampling from the most precise discretization, for a given level of error.  We cover extensions of results regarding variance and optimality criteria for the new approach. We apply the methodology to the problem of computing an unbiased  mollified version of the solution of a partial differential equation with random coefficients. A second application concerns the Bayesian inference (the smoothing problem) of an infinite dimensional signal  modelled by the solution of a stochastic partial differential equation that is observed on a discrete space grid and at discrete times. Both applications are complemented by numerical simulations.
\end{abstract}

\begin{keywords}
Monte Carlo, Multi-index discretization, Unbiasedness, Smoothing.
\end{keywords}

\section{Introduction}

Following Kalman's seminal paper \cite{kalman}, the theory and applications of filtering, smoothing and prediction has expanded rapidly and massively and is as strong today as it was seventy years ago. The area has evolved far beyond the original linear/Gaussian framework, many current applications involving nonlinear signals (possibly high dimensional or even infinite dimensional) and nonlinear observations (see \cite{cr} for some of the recent developments). General filtering, smoothing and prediction problems no longer have explicit solutions: inference is achieved mostly through numerical methods. Methods based on Monte Carlo integration are very popular, particularly for high dimensional problems.

Monte Carlo methods have proved to be sensible in a variety of applications. However, in certain contexts, multilevel Monte Carlo (MLMC) integration has been introduced with success,  see e.g. \cite{gile:08, giles:15}. This method is useful when one has a probability law associated to a discretization in a single dimension, for instance time. It can reduce the computational effort, relative to i.i.d.\ sampling (Monte Carlo) from the most precise discretization, for a given level of error. This has been extended in a subtle manner to discretizations in a multiple dimensions, for instance space and time, by \cite{Haji2016}. To be more precise let us assume that we are given a probability measure $\pi$ defined on a measure space $(E, \mathcal{E})$  and a collection of bounded-measurable real-valued functions $\mathcal{B}_b(E)$. We seek to compute for $\vphi \in \mathcal{B}_b(E)$
\eqns{
\bbE[\vphi(X)] = \int_E \vphi(x) \pi(\d x),
}
where $X$ is a random variable with law $\pi$. We assume that the random variable $X$ is associated to a $d$-dimensional multi-index continuum system, such as the one described in \cite{Haji2016}. For instance, $\pi$ may be associated to the solution of a stochastic partial differential equation (SPDE). Such systems are found in a wide variety of applications, see \cite{mcmc_diff} for examples. Alternatively $\pi$  can be the solution of a filtering/smoothing problem as we will explain later on.

Whilst the probability $\pi$ is not available to us, we have access to a class of \emph{biased} approximations $(\pi_{\bsalpha})_{\bsalpha\in\bbN^d_0}$ where $\bbN^d_0$ is the set on multi-indices of length $d$ with integer non-negative entries.  For instance, by using Monte Carlo integration, one can then compute $\bbE[\vphi(X_{\bsalpha})]$ where $X_{\bsalpha}$ being a random variable with law $\pi_{\bsalpha}$ and even though we have a bias, i.e., $\bbE[\vphi(X_{\bsalpha})] \neq \bbE[\vphi(X)]$ we will assume that 
\eqnl{eq:limitCondition}{
\lim_{\bsalpha \to \infty} | \bbE[\vphi(X_{\bsalpha})] - \bbE[\vphi(X)] | = 0,
}
where the limit $\bsalpha \to \infty$ is understood as $\min_{1 \leq i \leq d} \bsalpha_i \to \infty$ and, naturally,  that the computational cost associated with $\pi_{\bsalpha}$ increases as the values of $\bsalpha$ increase.

As an example of the above general context, consider discretely observing data associated to a signal modelled by the solution $(x_t)_{t\ge 0}$ of a stochastic partial differential equation (a concrete example can be found in \cref{ssec:SPDE}). Suppose that data is  obtained at unit time interval and is denoted by $y_1,\dots,y_K$. We assume that,  conditional upon $(x_t)_{t\ge 0}$, $y_1,\dots,y_K$ are conditionally independent with density on a finite dimensional space $g(y_k|x_k)$ with $x_k$ the solution of the SPDE at time $k$. Moreover we assume $g$ is well defined for any $x_k$, even if $x_k$ is discretized. Let $Q_{\bsalpha}$ be the transition density of the SPDE under a time and space discretization corresponding to a multi-index  $\bsalpha\in\bbN^d_0$. Given observed data $y_1,\dots,y_K$ our objective is to compute expectations with respect to the following distribution (known as a smoother or smoothing distribution for $x_{1:K}$): 
\eqns{
\pi_{\bsalpha}(x_{1:K}|y_{1:K}) \propto \prod_{k=1}^K g(y_k|x_k)Q_{\bsalpha}(x_{k}|x_{k-1})
}
where we assume that $x_0$ is given and we used the notation $x_{1:K}=(x_1,\dots,x_K)$ and $y_{1:K}=(y_1,\dots,y_K)$. If the class of discretizations is chosen  so that \cref{eq:limitCondition} holds, then the methodology developed in this paper can be applied to solve this problem. It is worth noting that, integrating with respect to  $\pi$ is non-trivial and well-known to be challenging. See, for instance, the work of \cite{beskos,kantas} and the references therein for more details.

\subsection{Contribution and Structure}

In the context described above, it is well-known that the Monte Carlo approximation of $\bbE[\vphi(X_{\bsalpha})]$, which we assume is required, can be significantly enhanced through the use of the MIMC method of \cite{Haji2016}. This idea is intrinsically linked to the popular MLMC approach. In this latter approach the dimensionality of the index is 1. Briefly writing the indices $1,\dots,L$, ($L$ being the finest discretization and 1 the coarsest and the discretization becomes more and more fine from 1) we have
$$
\bbE[\vphi(X_{L})] = \bbE[\vphi(X_{1})] + \sum_{l=2}^L\{\bbE[\vphi(X_{l})] - \bbE[\vphi(X_{l-1})]\}
$$
that is, introducing a collapsing sum representation of the expectation w.r.t.\ the finest discretization. The idea is then to \emph{dependently} sample a pair of approximations $(\pi_l,\pi_{l-1})$ (i.e~to dependently couple them) independently for $l\in\{2,\dots,L\}$. Note the case $l=1$ is just i.i.d.~sampling from $\pi_1$. If the dependence in the coupling is sufficiently positively correlated, then it is possible to sample fewer simulations at the fine discretizations (which are expensive to sample) and more at the coarse discretizations in such a way that the cost associated to obtain a prescribed mean square error of the MLMC approximation is less than that of i.i.d.\ sampling from the finest discretization. In the MIMC context, it is a more challenging procedure, but similar reductions in computational cost can also be possible.

A randomized version of the MLMC approach has been developed in \cite{Rhee2015}, which removed the discretization bias. In this work, we show how one can extend this idea in the context where the discretization parameters are in multiple dimensions. This extension allows for a judicious allocation of the computational effort in order to take into account the variance of the target distribution discretization in separate dimensions. In particular, Monte Carlo approximations are constructed to entirely remove the discretization bias, that is, to approximate $\bbE[\vphi(X)]$ directly. We also analyze the variance of the methodology and propose several original optimality criteria for its implementation. Several simulated examples are considered. 

This article is structured as follows. In \cref{sec:approach} we give some notations, the approach and some preliminary results. In \cref{sec:theory} our main theoretical results and the corresponding proofs are given. In \cref{sec:numerics} the new methodology is illustrated by numerical examples. \Cref{sec:summary} summarizes our work, with a discussion of future work.

\section{Notation and Preliminary Results}
\label{sec:approach}

Throughout the article, a complete probability space $(\Omega,\calF,\bbP)$ is considered with $\bbE$ denoting the expectation with respect to $\bbP$ and $\bbI_A$ denoting the random variable on $(\Omega,\calF,\bbP)$ defined as the indicator of the event $A \in \calF$. 

We work on the lattice $\bbN_0^d$ for some $d > 0$ equipped with the natural partial order $\leq$ which is defined as $\bsm \leq \bsn$ if and only if $\bsm_i \leq \bsn_i$ for all $1 \leq i \leq d$. Note that different total orders can also be defined on this $d$-dimensional lattice, such as the lexicographical order, but these total orders are to some extent arbitrary and will not be directly useful in the context we consider in this paper. For the sake of simplicity, $S$ and $S_{\bsalpha}$ will denote the random variables $\vphi(X)$ and $\vphi(X_{\bsalpha})$ respectively, for any $\bsalpha \in \bbN^d_0$. Let $\bsm,\bsn \in \bbN_0^d$ such that $\bsm \leq \bsn$ and consider
\eqns{
\calI_{\bsm}^{\bsn} \defeq \{ \bsalpha \in \bbN_0^d \st \bsm_1 \leq \bsalpha_1 \leq \bsn_1, \dots, \bsm_d \leq \bsalpha_d \leq \bsn_d \}.
}
By a usual abuse of notation, $\calI_{\bsm}^n$ is used instead of $\calI_{\bsm}^{\bsn}$ if the superscript $\bsn$ verifies $\bsn_i = n$ for all $1 \leq i \leq d$, and similarly for the subscript. It holds that $\calI^{\infty}_0 = \bbN^d_0$.

An estimator $\est{Z}$ of $\bbE S$ is defined as
\eqnl{eq:coupledSumFullN}{
\est{Z} \defeq \sum_{\bsalpha \in \calI_0^{\bsN}} \dfrac{\Delta S_{\bsalpha}}{\bbP(\bsN \geq \bsalpha)},
}
where $\bsN$ is a random variable on $\bbN_0^d$ independent of $(S_{\bsalpha})_{\bsalpha}$ which guarantees that the estimator is unbiased, i.e.\ it ensures that $\bbE \est{Z} = \bbE S$ holds, and where $\Delta \defeq \Delta_1\dots\Delta_d$ with
\eqns{
\Delta_i S_{\bsalpha} =
\begin{cases*}
S_{\bsalpha} - S_{\bsalpha - \bse_i} & if $\bsalpha_i > 0$ \\
S_{\bsalpha} & if $\bsalpha_i = 0$,
\end{cases*}
}
for any $i \in \{1,\dots,d\}$ where $\bse_i$ is the element of $\bbN^d_0$ such that $(\bse_i)_i = 1$ and $(\bse_i)_j = 0$ if $i \neq j$. The order of the application of the operators $\{\Delta_i\}_{i=1}^d$ in the definition of $\Delta$ is irrelevant since these operators can be easily seen to commute. The choice of a vector-valued random variable $\bsN$ in \cref{eq:coupledSumFullN} is justified by the fact that there might be interest in calculating the sum up to a non-diagonal index. For instance, relying on the increments with some very high and some very low components might yield an estimator with low variance at a reasonable computational cost. The following lemma will be useful to prove that the estimator $\est{Z}$ is unbiased.

\begin{lemma}
\label{lem:increment}
The increment $\Delta S_{\bsalpha}$ can be rewritten for any $\bsalpha \geq 1$ as
\eqnl{eq:lem:increment}{
\Delta S_{\bsalpha} = \sum_{\bsr \in \{0,1\}^d} (-1)^{|\bsr|} S_{\bsalpha-\bsr},
}
with $|\cdot|$ defined as the 1-norm on $\bbN_0^d$.
\end{lemma}

\begin{proof}
This result can be proved by recurrence on the dimension $d$. The case $d=1$ is obvious. If \cref{eq:lem:increment} is assumed to hold for a given $d \in \bbN$ then for any $\alpha' \in \bbN$,
\eqnsa{
\Delta_{d+1} \Delta_d \dots \Delta_1 S_{(\bsalpha,\alpha')} & = \Delta_d \dots \Delta_1 S_{(\bsalpha,\alpha')} - \Delta_d \dots \Delta_1 S_{(\bsalpha,\alpha'-1)} \\
& = \sum_{\bsr \in \{0,1\}^d} \sum_{r' \in \{0,1\}} (-1)^{|\bsr| + r'} S_{(\bsalpha-\bsr,\alpha'-r')} \\
& = \sum_{\bsr \in \{0,1\}^{d+1}} (-1)^{|\bsr|} S_{\bsbeta-\bsr}
}
where $\bsbeta = (\bsalpha,\alpha')$, hence showing that the relation is true for the dimension $d+1$.
\end{proof}

It follows from \cref{lem:increment} that a given term $S_{\bsalpha}$ is going to appear exactly once in each of the increments $\Delta S_{\bsalpha+\bsr}$ with $\bsr \in \{0,1\}^d$, that is $|\{0,1\}^d| = 2^d$ times, negatively in $2^{d-1}$ of the increments and positively in the other $2^{d-1}$ increments, therefore cancelling in $\est{Z}$. However, this does not take into account cases where the condition $\bsalpha \geq 1$ is not satisfied and is only valid for terms $S_{\bsalpha}$ for which all the $\{S_{\bsalpha+\bsr}\}_{\bsr \in \{0,1\}^d}$ are included in the considered sum.

\begin{lemma}
\label{lem:sumIncrements}
For any $k,n\in \bbN_0$ such that $n > k$, it holds that
\eqns{
\sum_{\bsalpha \in \calI_{k+1}^n} \Delta S_{\bsalpha} = \sum_{\bsalpha \in \{k,n\}^d} (-1)^{\ell_k(\bsalpha)} S_{\bsalpha},
}
where $\ell_{k}(\bsalpha)$ is the number of components of $\bsalpha$ equal to $k$. 
\end{lemma}

\begin{proof}
From \cref{lem:increment}, it holds that
\eqns{
\sum_{\bsalpha \in \calI_{k+1}^n} \Delta S_{\bsalpha} = \sum_{\bsr \in \{0,1\}^d} (-1)^{|\bsr|} \sum_{\bsalpha \in \calI_{k+1}^n} S_{\bsalpha-\bsr}.
}
The inner sum on the r.h.s.\ can be written as
\eqns{
\sum_{\bsalpha \in \calI_{k+1}^n} S_{\bsalpha-\bsr} = \sum_{I \subseteq \{1,\dots,d\}} \sum_{\substack{\bsalpha \in \calI_{k-\bsr+1}^{n-\bsr} \\ k < \bsalpha_i < n \iff i \in I}} S_{\bsalpha},
}
for any $\bsr \in \{0,1\}^d$. Denoting $|I|$ the cardinality of a set $I$, it follows that the terms $S_{\bsalpha}$ corresponding to a non-empty subset $I$ of $\{1,\dots,d\}$ appear $2^{|I|}$ times, $2^{|I|-1}$ times positively and $2^{|I|-1}$ times negatively, therefore cancelling out. The remaining terms are the ones corresponding to $I = \emptyset$ since they only appear $2^{|I|} = 1$ time. It follows that
\eqns{
\sum_{\bsalpha \in \calI_{k+1}^n} \Delta S_{\bsalpha} = \sum_{\bsr \in \{0,1\}^d} (-1)^{|\bsr|} S_{t_{k,n}(\bsr)},
}
where $t_{k,n} : \{0,1\}^d \to \{k,n\}^d$ is characterised by $(t_{k,n}(\bsr))_i = k^{\bsr_i}n^{1-\bsr_i}$ for any $i \in \{1,\dots,d\}$ and any $\bsr \in \{0,1\}^d$. The result of the lemma follows by a change of variable $\bsr \to \bsalpha = t_{k,n}(\bsr)$ in the sum on the r.h.s.\ and by verifying that
\eqns{
(t^{-1}_{k,n}(\bsalpha))_i = \frac{n-\bsalpha_i}{n-k} 
}
holds for any $i \in \{0,\dots,d\}$, so that $\big| t^{-1}_{k,n}(\bsalpha) \big| = \ell_k(\bsalpha)$. 

\end{proof}

\section{Main Theoretical Results}
\label{sec:theory}

We now consider several theoretical results for our approach, which justify its practical implementation.

\subsection{Unbiasedness}

\Cref{lem:sumIncrements} still does not apply to sums of increments containing indices $\bsalpha$ that do not verify $\bsalpha \geq 1$. Removing this last restriction leads in the following theorem.

\begin{theorem}
The estimator $\est{Z}$ is unbiased.
\end{theorem}

\begin{proof}
Let $\est{Z}_n$ be a partial version of the estimator $\est{Z}$ defined as
\eqns{
\est{Z}_n \defeq \sum_{\bsalpha \in \calI_0^{\bsN\land n}} \dfrac{\Delta S_{\bsalpha}}{\bbP(\bsN \geq \bsalpha)} = \sum_{\bsalpha \in \calI_0^n} \Delta S_{\bsalpha}\dfrac{\bbI_{\{\bsN \geq \bsalpha\}}}{\bbP(\bsN \geq \bsalpha)},
}
Because of the independence between $\bsN$ and the $\{S_{\bsalpha}\}_{\bsalpha}$, the estimator $\est{Z}_n$ satisfies
\eqns{
\bbE \est{Z}_n = \sum_{\bsalpha \in \calI_0^n} \bbE\Delta S_{\bsalpha} \bbE\bigg[\dfrac{\bbI_{\{\bsN \geq \bsalpha\}}}{\bbP(\bsN \geq \bsalpha)} \bigg]
= \sum_{\bsalpha \in \calI_0^n} \bbE\Delta S_{\bsalpha},
}
which can be further expressed as
\eqns{
\bbE \est{Z}_n = \sum_{I \subseteq \{1,\dots,d\}} \sum_{\substack{\bsalpha \in \calI_0^n \\ \bsalpha_i > 0 \iff i \in I}} \bbE\Delta_I S_{\bsalpha},
}
where, defining $k \defeq |I|$ and denoting $I = \{i_1,\dots,i_k\}$, the operator $\Delta_I$ is defined as $\Delta_I = \Delta_{i_1}\dots \Delta_{i_k}$. In the case $I = \emptyset$, the inner sum is equal to $S_{\zero}$. Using \cref{lem:sumIncrements}, it follows that
\eqns{
\bbE \est{Z}_n = \sum_{I \subseteq \{1,\dots,d\}} \sum_{\substack{\bsalpha \in \{0,n\}^d \\ \bsalpha_i > 0 \implies i \in I}} (-1)^{|(n-\bsalpha_I)/n|} \bbE S_{\bsalpha},
}
where, considering $\bsalpha$ as a function from $\{1,\dots,d\}$ to $\{0,\dots,n\}$, $\bsalpha_I$ denotes the restriction of $\bsalpha$ to the set $I$. Therefore, for any $\bsalpha \in \{0,1\}^d$, the term $\bbE S_{\bsalpha}$ appears once in the inner sum whenever the support $\supp(\bsalpha)$ of $\bsalpha$ is included in~$I$. Denoting $s(\bsalpha) \defeq |\supp(\bsalpha)|$, it holds that if $|I \setminus \supp(\bsalpha)| = l$ then $\bbE S_{\bsalpha}$ appears $\binom{d-s(\bsalpha)}{l}$ times, positively if $l$ is even and negatively if $l$ is odd. It follows that
\eqnsa{
\bbE \est{Z}_n & = \sum_{\bsalpha \in \{0,n\}^d} \bbE S_{\bsalpha} \bigg( \sum_{l=0}^{d-s(\bsalpha)} \binom{d-s(\bsalpha)}{l} (-1)^l \bigg) \\
& = \bbE S_n,
}
since the binomial formula in the first line differs from zero only when $\bsalpha \in \{0,n\}^d$ verifies $s(\bsalpha) = d$, that is when $\bsalpha = (n,n,\dots)$. The desired result follows by taking the limit under the condition stated in \cref{eq:limitCondition}.
\end{proof}

\subsection{Variance of the unbiased estimator}

Being able to determine the variance of the unbiased estimator will be important when looking for an optimal distribution for the random variable $\bsN$. We give expressions of the variance of $\est{Z}$ as well as for useful special cases. An additional notation is required for the statement of the following proposition: $\bsalpha \lor \bsbeta$ denotes the component-wise maximum of any $\bsalpha$ and $\bsbeta$ in $\bbN^d_0$, that is
\eqns{
\bsalpha \lor \bsbeta = (\bsalpha_1 \lor \bsbeta_1, \dots, \bsalpha_d \lor \bsbeta_d).
}

\begin{proposition}
\label{prop:varianceBarZ}
Assuming that
\eqnl{eq:prop:varianceBarZ:condition}{
\sum_{\bsalpha,\bsbeta \in \calI_0^{\infty}} \dfrac{\Delta \| S_{\bsalpha} - S \|_2 \Delta \| S_{\bsbeta} - S \|_2 }{\bbP(\bsN \geq \bsalpha \lor \bsbeta)} < \infty,
}
the second moment of $\est{Z}$ exists and is found to be
\eqnl{eq:generalSecondMoment}{
\bbE \est{Z}^2 = \sum_{\bsalpha,\bsbeta \in \calI_0^{\infty}} \est\nu_{\bsalpha,\bsbeta} \dfrac{\bbP(\bsN \geq \bsalpha \lor \bsbeta)}{\bbP(\bsN \geq \bsalpha) \bbP(\bsN \geq \bsbeta)}.
}
where $\est\nu_{\bsalpha,\bsbeta} = \bbE[\Delta S_{\bsalpha} \Delta S_{\bsbeta}]$.
\end{proposition}
\begin{remark}
The condition stated in \cref{eq:prop:varianceBarZ:condition} will hold if the probability $\bbP(\bsN \geq \bsalpha \lor \bsbeta)$ decreases sufficiently slowly when compared with the discretization error $\Delta \| S_{\bsalpha} - S \|_2 \Delta \| S_{\bsbeta} - S \|_2$. For instance, when solving a partial differential equation, the tail of the distribution of $\bsN$ should not be smaller than the decay of the error associated with the refinement of the mesh.
\end{remark}
\begin{proof}
In order to study the variance of the estimator, consider
\eqns{
\|\est{Z}_n - \est{Z}_k\|^2_2 = \sum_{\bsalpha,\bsbeta \in \calI_{k+1}^n} \bbE[\Delta S_{\bsalpha} \Delta S_{\bsbeta}] \dfrac{\bbP(\bsN \geq \bsalpha \lor \bsbeta)}{\bbP(\bsN \geq \bsalpha) \bbP(\bsN \geq \bsbeta)},
}
where $\|\cdot\|_2$ is the $L^2$-norm. For the same reasons as before, it can be verified that $\Delta S_{\bsalpha} = \Delta(S_{\bsalpha} - S)$ holds by adding and subtracting $2^{d-1}$ times the random variable $S$. It follows that
\eqns{
\|\est{Z}_n - \est{Z}_k\|^2_2 \leq \sum_{\bsalpha,\bsbeta \in \calI_{k+1}^n} \dfrac{\Delta \| S_{\bsalpha} - S \|_2 \Delta \| S_{\bsbeta} - S \|_2 }{\bbP(\bsN \geq \bsalpha \lor \bsbeta)},
}
where the inequality $\bbP(\bsN \geq \bsalpha) \bbP(\bsN \geq \bsbeta) \geq \bbP(\bsN \geq \bsalpha \lor \bsbeta)^2$, which holds for any $\bsalpha,\bsbeta \in \bbN^d_0$, has been used. Assuming that
\eqns{
\sum_{\bsalpha,\bsbeta \in \calI_0^{\infty}} \dfrac{\Delta \| S_{\bsalpha} - S \|_2 \Delta \| S_{\bsbeta} - S \|_2 }{\bbP(\bsN \geq \bsalpha \lor \bsbeta)} < \infty,
}
it follows that $\|\est{Z}_n - \est{Z}_k\|^2_2$ can be made arbitrarily small by considering $k$ large enough, so that $(\est{Z}_n)_n$ is a Cauchy sequence which therefore converges in the Hilbert space $L^2$, so that the second moment of $\est{Z}$ is finite. This completes the proof of the proposition.
\end{proof}

\begin{remark}
By considering a total order $\dotleq$ on $\bbN_0^d$ such as the lexicographical order, dual sums over $\bsalpha$ and $\bsbeta$ in some given subset of $\bbN^d_0$ could be split into diagonal and non-diagonal elements, the latter being simplified to terms verifying $\bsalpha \dotleq \bsbeta$. However, this would not allow for simplifying the indicator function $\bbI_{\{\bsN \geq \bsalpha \lor \bsbeta\}}$ since $\bsalpha \dotleq \bsbeta$ does not imply that $\bsalpha \lor \bsbeta = \bsbeta$ in general.
\end{remark}

\begin{remark}
\label{rem:oneDimensional}
The case $d = 1$ has been studied in \cite[Theorem~1]{Rhee2015} and yields an expression of the second moment $\bbE \est{Z}^2$ which simplifies drastically and which can be expressed as a single sum as
\eqns{
\bbE \est{Z}^2 = \sum_{\bsalpha \geq 0} \dfrac{ \|S_{\bsalpha-1} - S\|^2_2 - \|S_{\bsalpha} - S\|^2_2 }{\bbP(\bsN \geq \bsalpha)},
}
where $\bsalpha$ and $\bsN$ are now integers. The condition for the existence of $\bbE \est{Z}^2$ reduces to
\eqns{
\sum_{\bsalpha \geq 1} \dfrac{\| S_{\bsalpha-1} - S \|_2^2}{ \bbP(\bsN \geq \bsalpha) }
}
for $d = 1$, which is much simpler to verify than \cref{eq:prop:varianceBarZ:condition}.
\end{remark}

The random variable $\bsN$ can be chosen in such a way as to simplify \cref{eq:generalSecondMoment}: If the components of $\bsN$ are assumed to be independent random variables, then
\eqnsa{
\bbE \est{Z}^2 & = \sum_{\bsalpha,\bsbeta \in \calI_0^{\infty}} \est\nu_{\bsalpha,\bsbeta} \dfrac{\prod_{i=1}^d \bbP(\bsN_i \geq \bsalpha_i \lor \bsbeta_i)}{\prod_{i=1}^d \bbP(\bsN_i \geq \bsalpha_i) \bbP(\bsN \geq \bsbeta_i)} \\
& = \sum_{\bsalpha,\bsbeta \in \calI_0^{\infty}} \dfrac{\est\nu_{\bsalpha,\bsbeta}}{\bbP(\bsN \geq \bsalpha \land \bsbeta)}.
}
However, this expression of the second moment is still more complicated than in the case $d=1$ detailed in \cref{rem:oneDimensional} as it involves a double sum. Yet, another special case of the estimator $\est{Z}$ can be obtained by assuming that $\bsN$ is a constant function, i.e.\ that $\bsN_i = \bsN_j$ almost surely for any $i,j \in \{1,\dots,d\}$. This estimator will be denoted $\est{Z}'$ and is expressed as follows
\eqns{
\est{Z}' = \sum_{\bsalpha \in \calI_0^N} \dfrac{\Delta S_{\bsalpha}}{\bbP(N \geq \supn{\bsalpha})},
}
where $N$ is the integer-valued random variable induced by $\bsN = (N,N,\dots)$ and where $| \cdot |_{\infty}$ is the supremum norm. Since $\est{Z}'$ is defined as the estimator $\est{Z}$ for a special choice for $\bsN$, it is also unbiased.

\begin{proposition}
\label{prop:varianceBarZPrime}
If it exists, the second moment of the estimator $\est{Z}'$ takes the form
\eqns{
\bbE \est{Z}'^2 = \sum_{\bsalpha \in \calI_0^{\infty}} \dfrac{\est\nu'_{\bsalpha}}{\bbP(N \geq \supn{\bsalpha})},
}
with $\est\nu'_{\bsalpha} \defeq \bbE\big[\Delta S_{\bsalpha} \big( (S - S_{\supn{\bsalpha}-1}) + (S - S_{\supn{\bsalpha}})\big)\big]$.
\end{proposition}
\begin{remark}
The expression of the second moment $\bbE\est{Z}'^2$ is closer to the one obtained in \cref{rem:oneDimensional} for the case $d=1$. This is natural since the simplification from $\est{Z}$ to $\est{Z}'$ amounts to making $\bsN$ single-variate, so that only the terms $\{S_{\bsalpha}\}_{\bsalpha}$ retain their multi-index nature. The expression of $\est\nu'_{\bsalpha}$ for $d=1$ can be recovered easily as
\eqnsa{
\est\nu'_{\bsalpha} & = \bbE\big[ ( S_{\bsalpha} - S_{\bsalpha-1} )\big((S - S_{\bsalpha-1}) + (S - S_{\bsalpha})\big) \big] \\
& = \bbE\big[ \big((S - S_{\bsalpha-1}) - (S - S_{\bsalpha})\big)\big((S - S_{\bsalpha-1}) + (S - S_{\bsalpha})\big) \big] \\
& = \|S - S_{\bsalpha-1}\|_2^2 - \|S - S_{\bsalpha}\|_2^2.
}
\end{remark}
\begin{proof}
As in the proof of \cref{prop:varianceBarZ}, a partial version of $\est{Z}'$ can be introduced as
\eqns{
\est{Z}'_n \defeq \sum_{\bsalpha \in \calI_0^{N\land n}} \dfrac{\Delta S_{\bsalpha}}{\bbP(N \geq \supn{\bsalpha})} = \sum_{\bsalpha \in \calI_0^n} \Delta S_{\bsalpha}\dfrac{\bbI_{\{N \geq \supn{\bsalpha} \}}}{\bbP(N \geq \supn{\bsalpha})}.
}
It holds that
\eqnsa{
\bbE \est{Z}'^2_n & = \sum_{\substack{\bsalpha,\bsbeta \in \calI_0^n \\ \supn{\bsalpha} = \supn{\bsbeta}}} \dfrac{\bbE[\Delta S_{\bsalpha} \Delta S_{\bsbeta}]}{\bbP(N \geq \supn{\bsalpha})} + 2\sum_{\substack{\bsalpha,\bsbeta \in \calI_0^n \\ \supn{\bsalpha} < \supn{\bsbeta}}} \dfrac{\bbE[\Delta S_{\bsalpha} \Delta S_{\bsbeta}]}{\bbP(N \geq \supn{\bsalpha})} \\
& = \sum_{\bsalpha \in \calI_0^n} \dfrac{ \bbE\big[ \Delta S_{\bsalpha} ( S_{|\bsalpha|_{\infty}} - S_{|\bsalpha|_{\infty}-1} )  + 2 \Delta S_{\bsalpha} ( S_n - S_{|\bsalpha|_{\infty}} ) \big] }{\bbP(N \geq \supn{\bsalpha})},
}
where the following relations have been used with $m = |\bsalpha|_{\infty} \leq n$:
\eqns{
\sum_{\substack{\bsbeta \in \calI_0^n \\ \supn{\bsbeta} > m }} \Delta S_{\bsbeta} = \sum_{\bsbeta \in \calI_0^n} \Delta S_{\bsbeta} - \sum_{\bsbeta \in \calI_0^m} \Delta S_{\bsbeta} = S_n - S_m
}
and
\eqns{
\sum_{\substack{\bsbeta \in \calI_0^n \\ \supn{\bsbeta} = m }} \Delta S_{\bsbeta} = \sum_{\substack{\bsbeta \in \calI_0^m \\ \supn{\bsbeta} > m-1 }} \Delta S_{\bsbeta} = S_m - S_{m-1}.
}
The desired result is obtained by rearranging the terms and taking the limit $n \to \infty$.
\end{proof}

We note that if one produces independent realizations of $Z'$ then \cref{prop:varianceBarZPrime} can be used to obtain a specific variance. That is, for specific models (see e.g.\ \cite{Haji2016}) one has expressions for $\est\nu'_{\bsalpha}$, appropriately centered, in terms of a function $\psi(\bsalpha)$ which goes to zero as $\min_i \alpha_i\rightarrow \infty$ such that $\bbE \est{Z}'^2<+\infty$. Then, for some $\epsilon>0$, one can choose the number of samples to make the variance $\mathcal{O}(\epsilon^2)$ as in the MLMC/MIMC literature \cite{gile:08,Haji2016}. 

Following \cite{Rhee2015}, a variant of the estimator $Z$ can be introduced as follows
\eqns{
\tilde{Z} = \sum_{\bsalpha \in \calI_0^{\bsN}} \dfrac{\tilde\Delta_{\bsalpha}}{\bbP(\bsN \geq \bsalpha)}
}
where the random variable $\tilde\Delta_{\bsalpha}$ is defined for any $\bsalpha \in \bbN_0^d$ as
\eqns{
\tilde\Delta_{\bsalpha} \defeq \sum_{\bsr \in \{0,1\}^d} (-1)^{|\bsr|} \tilde{S}_{\bsalpha-\bsr}
}
with the joint random variable $(\tilde{S}_{\bsalpha-\bsr})_{\bsr \in \{0,1\}^d}$ having the same marginal distributions as the joint $(S_{\bsalpha-\bsr})_{\bsr \in \{0,1\}^d}$, making $\tilde{Z}$ unbiased. The estimators $Z$ and $\tilde{Z}$ can be respectively referred to as the \emph{coupled-sum} estimator and the \emph{independent-sum} estimator. A simpler version of the estimator $\tilde{Z}$ can be introduced as previously for $Z$ by assuming that realisations of $\bsN$ are constant on $\bbN_0^d$ almost surely:
\eqns{
\tilde{Z}' = \sum_{\bsalpha \in \calI_0^N} \dfrac{\tilde\Delta_{\bsalpha}}{\bbP(N \geq \supn{\bsalpha})}.
}

\begin{proposition}
If it exists, the second-moment of the estimator $\tilde{Z}'$ is found to be
\eqns{
\bbE \tilde{Z}'^2 = \sum_{\bsalpha \in \calI_0^{\infty}} \dfrac{\tilde\nu'_{\bsalpha}}{\bbP(N \geq \supn{\bsalpha})},
}
with $\tilde\nu'_{\bsalpha} \defeq \var(\Delta S_{\bsalpha}) + \bbE\Delta S_{\bsalpha} \big( (\bbE S - \bbE S_{\supn{\bsalpha}-1}) + (\bbE S - \bbE S_{\supn{\bsalpha}})\big)$.
\end{proposition}

\begin{proof}
The partial version $\tilde{Z}'_n$ of the estimator $\tilde{Z}'$ is introduced as before with $Z'$ and verifies
\eqnsa{
& \bbE \tilde{Z}'^2_n \\
& = \sum_{\bsalpha \in \calI_0^n} \dfrac{\bbE[\tilde\Delta_{\bsalpha}^2]}{\bbP(N \geq \supn{\bsalpha})} + \sum_{\substack{\bsalpha,\bsbeta \in \calI_0^n \\ \supn{\bsalpha} = \supn{\bsbeta} \\ \bsbeta \neq \bsalpha}} \dfrac{\bbE \tilde\Delta_{\bsalpha} \bbE \tilde\Delta_{\bsbeta}}{\bbP(N \geq \supn{\bsalpha})} + 2\sum_{\substack{\bsalpha,\bsbeta \in \calI_0^n \\ \supn{\bsalpha} < \supn{\bsbeta}}} \dfrac{\bbE\tilde\Delta_{\bsalpha} \bbE\tilde\Delta_{\bsbeta}}{\bbP(N \geq \supn{\bsalpha})} 
}
\eqnsa{
& = \sum_{\bsalpha \in \calI_0^n} \dfrac{ \bbE\tilde\Delta_{\bsalpha}^2 + \bbE\tilde\Delta_{\bsalpha} ( \bbE S_{|\bsalpha|_{\infty}} - \bbE S_{|\bsalpha|_{\infty}-1} - \bbE\tilde\Delta_{\bsalpha} )  + 2 \bbE\tilde\Delta_{\bsalpha} ( \bbE S_n - \bbE S_{|\bsalpha|_{\infty}} ) \big] }{\bbP(N \geq \supn{\bsalpha})} \\
& = \sum_{\bsalpha \in \calI_0^n} \dfrac{\tilde\nu'_{\bsalpha}}{\bbP(N \geq \supn{\bsalpha})},
}
from which the result of the proposition follows.
\end{proof}

\subsection[Optimal distribution for N]{Optimal distribution for $\bsN$}
\label{ssec:optimalDistributionN}

Since $\bsN$ is a design random variable, its distribution can be chosen in a way that maximises the performance of the corresponding MIMC method in a sense to be defined. The objective is to optimise jointly the computational effort and the accuracy of the algorithm. The former is quantified by the time necessary to compute one instance of $\est{Z}$ while the latter is represented by the variance of $\est{Z}$. Consider an arbitrary total order $\dotleq$ on $\bbN_0^d$ that is compatible with $\leq$, i.e.\ such that $\bsalpha \leq \bsbeta$ implies $\bsalpha \dotleq \bsbeta$ for any $\bsalpha,\bsbeta \in \bbN^d_0$, and define $t_{\bsalpha}$ as the time to compute the terms in $\Delta S_{\bsalpha}$ that have not already been computed for previous $\Delta S_{\bsbeta}$ with $\bsbeta \dotle \bsalpha$. Let $\est\tau$ be the time necessary to compute $\est{Z}$, then
\eqns{
\bbE\est\tau = \sum_{\bsalpha \in \calI_0^{\infty}} t_{\bsalpha} \bbP(\bsN \geq \bsalpha).
}
For a given duration $c$, let $n_c \defeq \max\{ n \in \bbN_0 \st \sum_{i=1}^n \est\tau_i \leq c \}$ be the number of copies $\est{Z}_i$ of $\est{Z}$ that can be generated in this amount of time. The sample average $\est{m}_n$ defined as
\eqns{
\est{m}_n = \dfrac{1}{n} \sum_{i = 1}^n \est{Z}_i,
}
can then be used to formulate a CLT when $\bbE\est\tau$ and $\var\est{Z}$ are finite as \cite{Glynn1992}
\eqnl{eq:CLT}{
c^{1/2} (\est{m}_{n_c} - m) \implies \big( \bbE[ \est\tau] \var\est{Z} \big)^{1/2} \calN(0,1),
}
where $m \defeq \bbE S$, where $\calN(0,1)$ is the standard Gaussian distribution and where $\implies$ denotes the convergence in probability when $c \to \infty$.

\begin{remark}
For instance, the computational time $t'_{\bsalpha}$ for the term $S_{\bsalpha}$ can be assumed to be equal to $2^{|\bsalpha|}$. This assumption makes sense in many cases including the ones considered here in the numerical results where partial differential equations are solved on meshes with $2^{|\bsalpha|}$ elements. If we consider the independent sum-estimator then $t_{\bsalpha}$ is the computational time for the whole of $\Delta S_{\bsalpha}$ which verifies
\eqns{
t_{\bsalpha} = \sum_{\bsr \in \{0,1\}^d} (-1)^{|\bsr|} t'_{\bsalpha - \bsr} \leq 2^{|\bsalpha| + d}.
}
Then, the expected computational time $\bbE[\tau]$ will be finite if the probability $\bbP(\bsN \geq \bsalpha)$ is of order $O(2^{-r|\bsalpha|})$ with $r > 1$. However, for \cref{eq:prop:varianceBarZ:condition} to hold, the tail of the distribution of $\bsN$ also needs to be sufficiently large. For instance, if $\| S_{\bsalpha} - S \| = O(2^{-|\bsalpha|p})$ for some $p > 0$ related to the considered numerical scheme, then
\eqns{
\Delta \| S_{\bsalpha} - S \|_2 \Delta \| S_{\bsbeta} - S \|_2 = O(2^{-p|\bsalpha+\bsbeta|})
}
and $\bbP(\bsN \geq \bsalpha \lor \bsbeta) = O(2^{-r|\bsalpha \lor \bsbeta|})$ also have to verify $r < p$ since $|\bsalpha \lor \bsbeta| \leq |\bsalpha + \bsbeta|$ for any $\bsalpha,\bsbeta \in \bbN^d_0$. The condition \cref{eq:prop:varianceBarZ:condition} can be weakened for special cases to allow for more freedom in the choice of $\bsN$.
\end{remark}

\Cref{eq:CLT} indicates that the distribution of the random variable $\bsN$ can be chosen in a way to make the product between the expected computational time $\bbE\est\tau$ and the variance $\var\est{Z}$ as small as possible. The following problem is therefore considered:
\eqnla{eq:optimalDistributionFormulation}{
& \min_{\est{F}} \est{g}(\est{F}) \defeq \Bigg( \sum_{\bsalpha,\bsbeta \in \calI_0^{\infty}} \est\nu_{\bsalpha,\bsbeta} \dfrac{\est{F}_{\bsalpha \lor \bsbeta}}{\est{F}_{\bsalpha} \est{F}_{\bsbeta}} - m^2 \Bigg) \Bigg( \sum_{\bsalpha \in \calI_0^{\infty}} t_{\bsalpha} \est{F}_{\bsalpha} \Bigg) \\
\label{eq:optimalDistributionFormulation_constraint1}
& \text{subject to } (\est{F}_{\bsalpha})_{\bsalpha} \text{ is a strictly positive net} \\
\label{eq:optimalDistributionFormulation_constraint2}
& \hphantom{\text{subject to }} \est{F}_{\bsalpha} \geq \est{F}_{\bsbeta} \text{ for any } \bsalpha \leq \bsbeta \\
\label{eq:optimalDistributionFormulation_constraint3}
& \hphantom{\text{subject to }} \est{F}_0 = 1.
}
The solution to \cref{eq:optimalDistributionFormulation} is difficult to formulate in general, however, the special case of the estimator $\est{Z}'$ yields the simpler problem:
\eqns{
\min_{\est{F}} \est{g}'(\est{F}) \defeq \Bigg( \sum_{\bsalpha \in \calI_0^{\infty}} \dfrac{\est\nu'_{\bsalpha}}{\est{F}_{\bsalpha}} - m^2 \Bigg) \Bigg( \sum_{\bsalpha \in \calI_0^{\infty}} t_{\bsalpha} \est{F}_{\bsalpha} \Bigg),
}
subject to \crefrange{eq:optimalDistributionFormulation_constraint1}{eq:optimalDistributionFormulation_constraint3} and
\eqnl{eq:optimalDistributionFormulation_constraint4}{
\est{F}_{\bsalpha} = \est{F}_{\bsbeta} \text{ for any } \bsalpha,\bsbeta \text{ verifying } |\bsalpha|_{\infty} = |\bsbeta|_{\infty}.
}
By a direct generalisation of \cite[Proposition~1]{Rhee2015}, it holds that if the net $(\est\mu'_{\bsalpha})_{\bsalpha}$, defined as $\est\mu'_{\zero} = \est\nu'_{\zero}-m^2$ and $\est\mu'_{\bsalpha} = \est\nu'_{\bsalpha}$ for any $\bsalpha \neq \zero$, is non-negative then the following inequality holds
\eqns{
\est{g}'(\est{F}) \geq \bigg( \sum_{\bsalpha \in \calI_0^{\infty}} \sqrt{\est\mu'_{\bsalpha} t_{\bsalpha}} \bigg) = \est{g}'(\est{F}^{\dagger})
}
with $\est{F}^{\dagger}$ characterised by
\eqns{
\est{F}^{\dagger}_{\bsalpha} \defeq \dfrac{\sqrt{\est\mu'_{\bsalpha} / t_{\bsalpha}}}{\sqrt{\est\mu'_{\zero} / t_{\zero}}},
}
for any $\bsalpha \in \bbN^d_0$. However, $\est{F}^{\dagger}$ might not be feasible and the solution over the feasible region is denoted $\est{F}^*$; by \cite[Proposition~2]{Rhee2015}, this minimum is achieved if $(\mu_{\bsalpha})_{\bsalpha}$ is positive and $(t_{\bsalpha})_{\bsalpha}$ is bounded below by a positive constant.

Due to the constraint \cref{eq:optimalDistributionFormulation_constraint4}, the probability mass function induced by $\est{F}^*$ can be characterised by its values on the diagonal of $\bbN_0^d$ and we consider the sequence of indices $J^* = (i^*_j)_{j \geq 0}$ such that $i^*_0 = 0$ and
\eqns{
i^*_j = \inf\{ k > i^*_{j-1} \st \est{F}^*_k < \est{F}^*_{i^*_{j-1}} \},
}
where $\est{F}^*_k$ is a shorthand notation for $\est{F}^*_{\bsk}$ with $\bsk = (k,\dots,k)$. Denoting, for any strictly increasing integer-valued sequence $J = (i_j)_{j \geq 0}$,
\eqns{
\est\mu'_j(J) = \sum_{\substack{\bsalpha \in \calI_0^{\infty} \\ i_j \leq |\bsalpha|_{\infty} < i_{j+1}}} \est\mu'_{\bsalpha} \AND t_j(J) = \sum_{\substack{\bsalpha \in \calI_0^{\infty} \\ i_j \leq |\bsalpha|_{\infty} < i_{j+1}}} t_{\bsalpha},
}
it follows that
\eqns{
\est{g}'(\est{F}^*) = \bigg( \sum_{k \geq 0} \dfrac{\est\mu'_k(J^*)}{\est{F}^*_{i^*_k}} \bigg) \bigg( \sum_{k \geq 0} t_k(J^*) \est{F}^*_{i^*_k} \bigg).
}
Extending the results of \cite[Theorem~3]{Rhee2015} to the considered setting, it holds that if $(\est\mu'_{\bsalpha})_{\bsalpha}$ is a positive net and $(t_{\bsalpha})_{\bsalpha}$ is non-decreasing w.r.t.\ $\dotleq$, then there exists an optimiser $\est{F}^*$ inducing a sequence $J^*$ such that
\eqns{
\est{F}^*_j = \sqrt{ \dfrac{\est\mu'_{\gamma_j}(J^*) / t_{\gamma_j}(J^*)}{\est\mu'_0(J^*) / t_0(J^*)} },
}
where $\gamma_j$ is the unique integer verifying $i^*_{\gamma_j} \leq j < i^*_{\gamma_j+1}$. It follows that
\eqns{
\est{g}'(\est{F}^*) = \bigg(  \sum_{k \geq 0} \sqrt{ \est\mu'_k(J^*) / t_k(J^*) } \bigg)^2.
}
These expressions are the same for unbiased MLMC and the considered instance of unbiased MIMC so that \cite[Algorithm~1]{Rhee2015} can be used to find the desired optimal sequence $J^*$.

\section{Numerical Results}
\label{sec:numerics}

In order to evaluate the performance of the proposed unbiased MIMC (UMIMC) method, a comparison with the MIMC algorithm of \cite{Haji2016} is performed on two different problems. The first is covered in \cref{ssec:PDE} and comprises computing a mollified version of the solution of a partial differential equation with random coefficients. The second application is an inference problem for a partially observed signal modelled by an SPDE on a 1-dimensional domain in \cref{ssec:SPDE}.
We begin by giving some implementation details for the UMIMC method.

\subsection{Implementation}
\label{sec:implementation}

In this section as well as in the numerical results, we make use of the simplified version $\tilde{Z}'$ of the independent-sum estimator $\tilde{Z}$. Since, in practice, realisations of $S_{\bsalpha}$ can only be computed up to a certain level, the partial estimator $\tilde{Z}'_m$ defined as
\eqns{
\tilde{Z}'_m \defeq \sum_{\bsalpha \in \calI^{N\land m}_0} \dfrac{\tilde\Delta_{\bsalpha}}{\bbP(N \geq \supn{\bsalpha})}
}
has to be considered instead, for a given $m$. In order to accurately calculate the optimal sequence~$J^*$ yielding the optimal distribution for $N$ as described in \cref{ssec:optimalDistributionN}, a high number of realisations of $\tilde\Delta_{\bsalpha}$ has to be computed for any $\bsalpha \in \calI^m_0$. This fact implies that the computational effort required to start estimating from $\tilde{Z}'_m$ is high. To bypass this limitation and reduce the number of realisations of $\tilde\Delta_{\bsalpha}$ computed before calculating $J^*$, the latter is updated frequently. One of the consequences of this solution is that the probabilities $\bbP(N \geq \supn{\cdot})$ will vary through the iterations. This drawback can be compensated for by dividing the increment $\tilde\Delta_{\bsalpha}$ by the number of times it has been sampled, instead of the probability of it being sampled. In spite of their difference, these two normalisations are equivalent asymptotically, by an easy application of the strong law of large numbers. 
Note that with the adaptations, one must do some more work to verify the consistency of the estimator, but this is left for future work.
The estimation thus takes the following form in practice:
\eqns{
\bbE\tilde{Z}'_m \approx \sum_{i=1}^M \bigg[ \sum_{\bsalpha \in \calI^{n_i \land m}_0} \dfrac{\tilde\Delta_{\bsalpha,i}}{ \sum_{j=1}^M \bbI_{\{n_j \geq \supn{\bsalpha}\}} } \bigg],
}
where $M$ is the total number of iterations, where $\tilde\Delta_{\bsalpha,i}$ is the sampled increment and where $n_i$ is a sample from $N$ at the $i$\textsuperscript{th} iteration.

\subsection{A Partial Differential Equation with random coefficients}
\label{ssec:PDE}

We consider here a partial differential equation with random coefficients of the form
\begin{subequations}
\label{eq:PDE}
\eqnl{eq:PDE_domain}{
-\nabla \cdot (a(\bsx;\omega) \nabla u(\bsx;\omega)) = 1 \qquad \text{for }\bsx \in D
}
on a domain $D = [0,1]^2$ with $\omega\in\Omega$. We will assume that $u$ satisfies the Dirichlet boundary condition
\eqnl{eq:PDE_boundary}{
u(\bsx;\omega) = 0 \qquad \text{for }\bsx \in \partial D. 
}
\end{subequations}
Similarly to \cite{Haji2016}, the diffusion coefficient $a$ is defined as
\eqns{
a(\bsx;\omega) = 1 + \exp\big(2 Y_1(\omega) \sin(\pi \bsx_1)\cos(\pi \bsx_2) + 2 Y_2(\omega) \cos(4\pi\bsx_1)\sin(4\pi\bsx_2)\big),
}
and the random variable of interest is
\eqns{
X(\omega) = \dfrac{100}{\sigma\sqrt{2\pi}} \int_D \exp\bigg( -\dfrac{\| \bsx - \bsx_0 \|^2_2}{2\sigma^2}\bigg) u(\bsx;\omega) \d\bsx,
}
with $\sigma = 0.16$ and $\bsx_0 = [0.5; 0.2]$. The objective is to compute $\bbE[X]$ using the proposed method with $d=2$. For a each realisation of $Y_1$ and $Y_2$, the partial differential equation \cref{eq:PDE} is solved by finite element method on a linear and uniform meshing defined on $D$ with $4\times 2^{\bsalpha_i}$ elements in the $i$\textsuperscript{th} dimension for a multi-index $\bsalpha$. The terms corresponding to any index $\bsalpha$ such that $|\bsalpha_2 -\bsalpha_1| > 2$ are not computed in order to avoid numerical issues with degenerated elements.

To better understand the accuracy associated with each index, some of the produced meshes at different levels are shown in \cref{fig:solutionAtLevels}. \Cref{fig:TOLvsWork} shows the RMSE as a function of the computational effort for the unbiased MIMC when the greatest multi-index available $\bsalpha_{\max} = \alpha_{\max}(1,1)$ is either $(3,3)$, $(4,4)$ or $(5,5)$, compared with the MIMC algorithm described in \cite[Sec.\ 3.2.2]{Haji2016}. Three scalar parameters have to be set in the MIMC algorithm, the accuracy $\mathrm{TOL}$, a splitting parameter $\theta$ and a confidence level $0 < \epsilon \ll 1$ defined such that
\eqnsa{
| \bbE[ Z_{\text{\sc mimc}} - S ] |  &\leq (1-\theta)\mathrm{TOL} \\
\bbP( | Z_{\text{\sc mimc}} - \bbE[Z_{\text{\sc mimc}}] | \leq \theta\mathrm{TOL} ) & \geq 1-\epsilon,
}
where $Z_{\text{\sc mimc}}$ is the MIMC estimator. The values $\mathrm{TOL} = 5\times 10^{-3}$, $\theta = 0.5$ and $\epsilon = 0.25$ are considered here. The maximum computational effort considered is increased with the value of $\alpha_{\max}$ to take into account the corresponding computational overhead.

Comparing \crefrange{fig:TOLvsWork_3}{fig:TOLvsWork_5}, it appears that the proposed implementation of the UMIMC and the considered version of MIMC behave very differently in time: the UMIMC has a higher error at the start since it relies on all levels at all times and requires more computational effort to compensate for the randomness in the coefficient of the considered PDE. This effect is more pronounced when the value of $\alpha_{\max}$ increases. However, the error in the MIMC increases at the times when it starts to perform computations for higher indices, since the effect of the random coefficients has to be averaged again, whereas the error of the UMIMC decreases monotonically. These remarks about each method do not allow to conclude that one is better than the other, but they highlights the differences in terms of implementation: the considered version of the MIMC attempts to reach a given level of precision determined by some parameters while the UMIMC requires the setting of fewer parameters but offers less control on its behaviour.

\begin{figure}
\thispagestyle{empty}
\centering
\includegraphics[trim=90pt 280pt 75pt 95pt, clip,width=\textwidth]{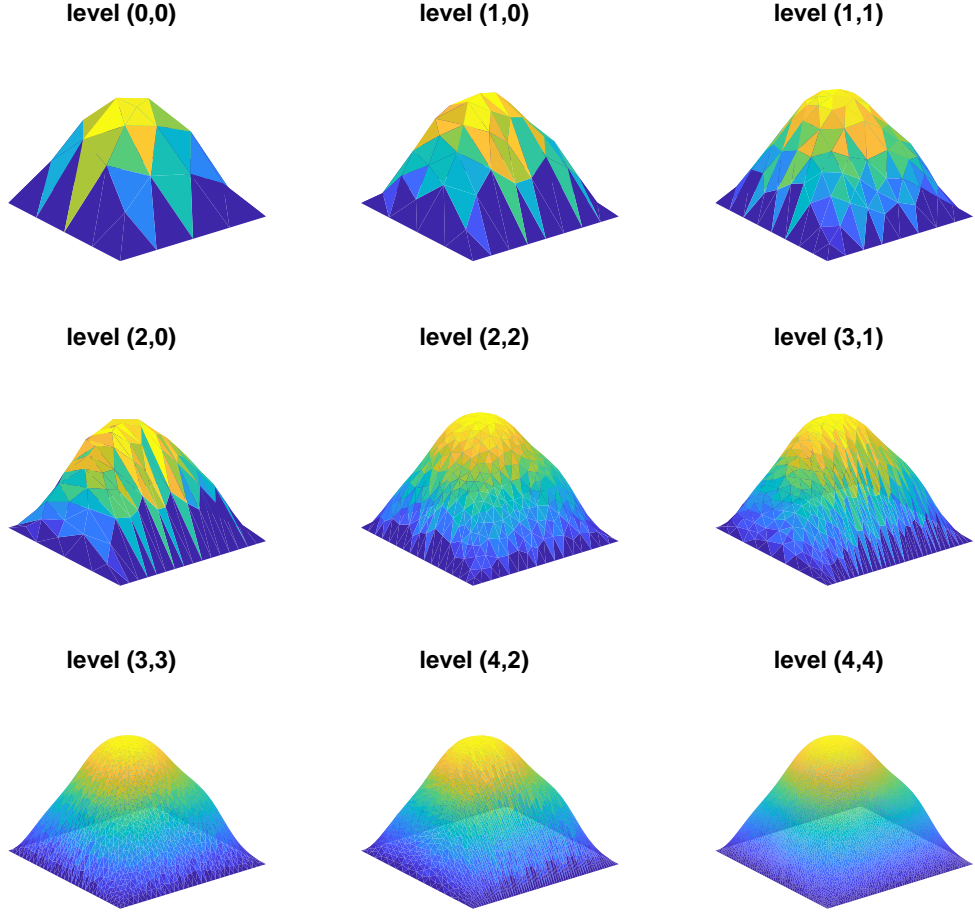}
\caption{Solution of \cref{eq:PDE} obtained when letting $Y_1 = Y_2 = 1/2$ a.s., at different levels.}
\label{fig:solutionAtLevels}
\end{figure}

\begin{figure}
\thispagestyle{empty}
\centering
\begin{subfigure}[b]{0.72\textwidth}
\includegraphics[width=\textwidth,trim = 105pt 280pt 90pt 290pt, clip]{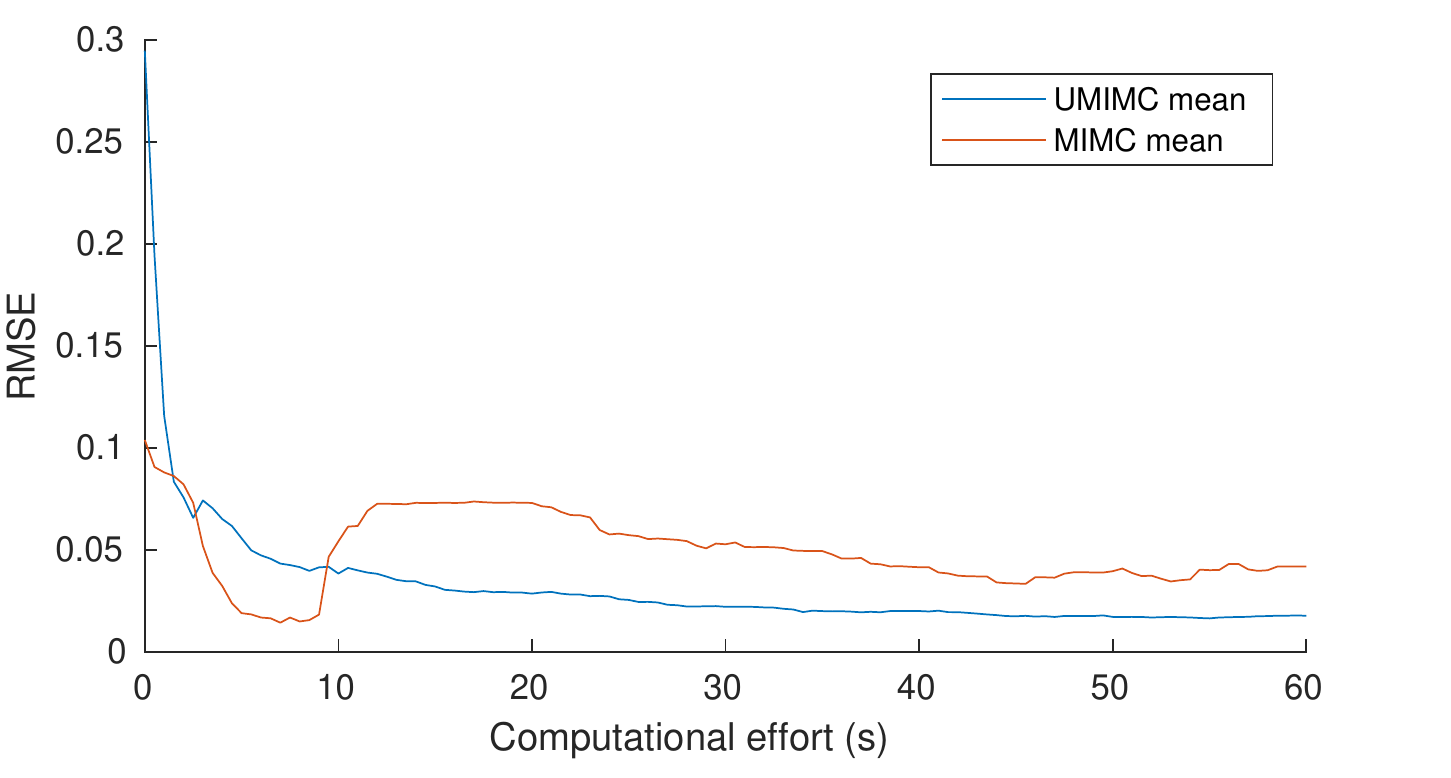}
\caption{Maximum level $(3,3)$}
\label{fig:TOLvsWork_3}
\end{subfigure}
\begin{subfigure}[b]{0.72\textwidth}
\includegraphics[width=\textwidth,trim = 105pt 280pt 90pt 290pt, clip]{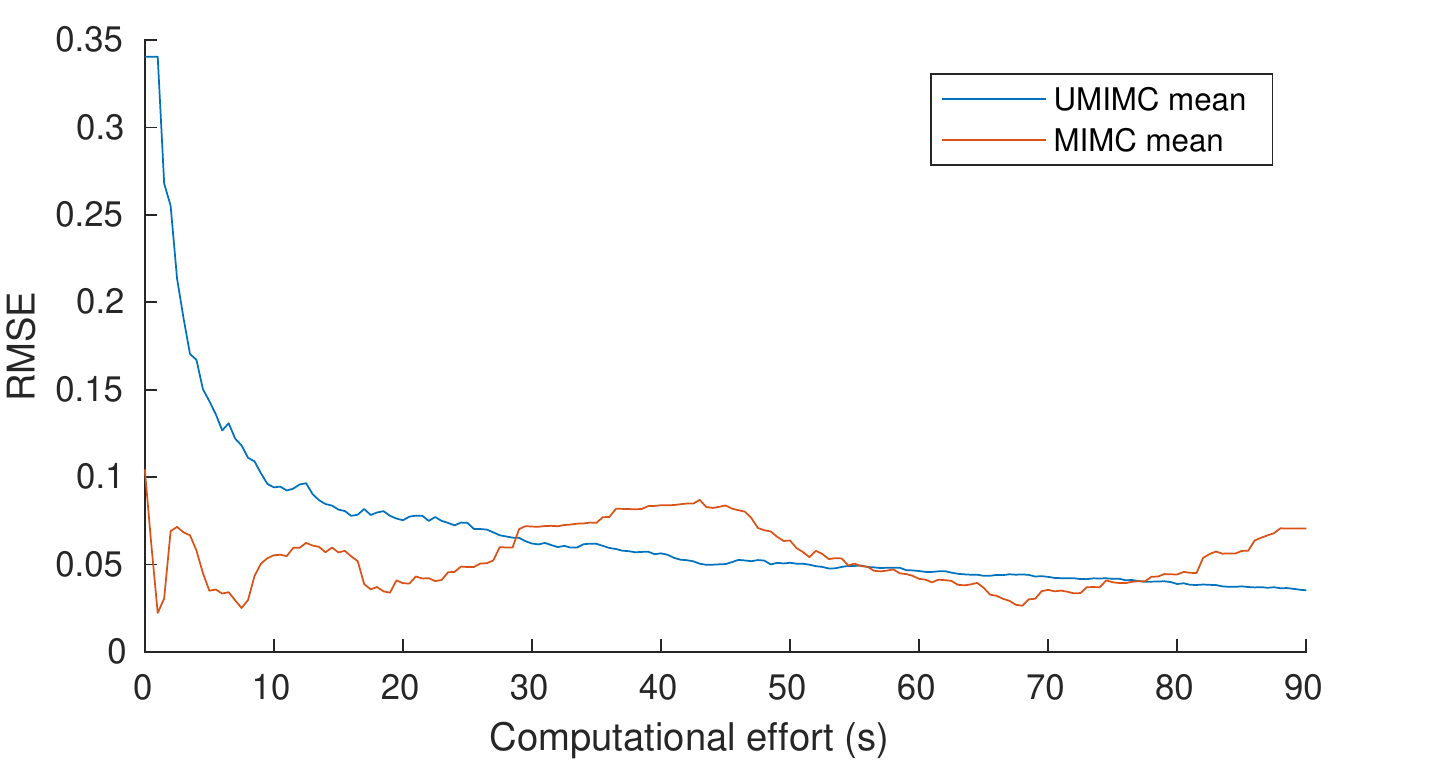}
\caption{Maximum level $(4,4)$}
\label{fig:TOLvsWork_4}
\end{subfigure}
\begin{subfigure}[b]{0.72\textwidth}
\includegraphics[width=\textwidth,trim = 105pt 280pt 90pt 290pt, clip]{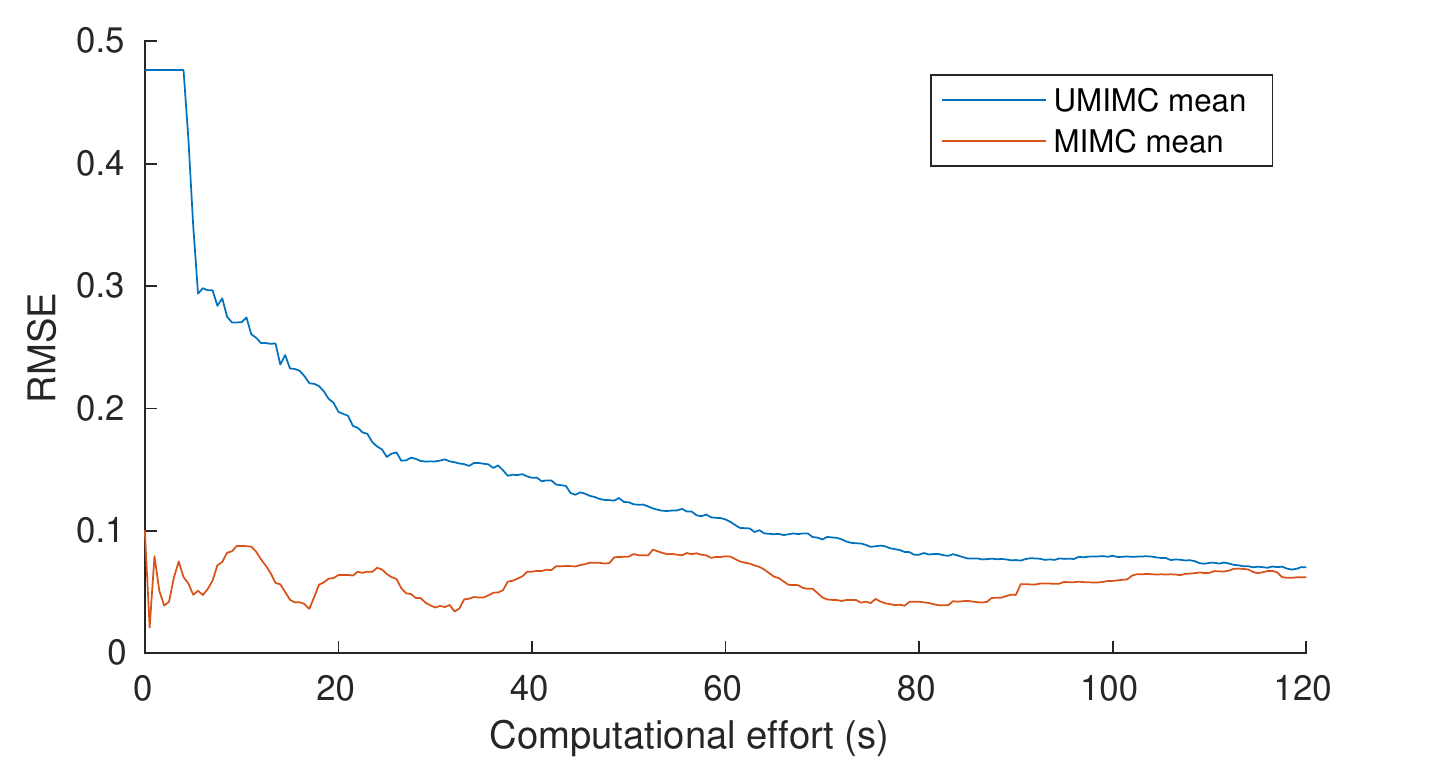}
\caption{Maximum level $(5,5)$}
\label{fig:TOLvsWork_5}
\end{subfigure}
\caption{RMSE of the MIMC and UMIMC algorithms as a function of the computational effort for the problem \cref{eq:PDE}. The results are averaged over $50$ Monte Carlo runs.}
\label{fig:TOLvsWork}
\end{figure}

\subsection{Inference of partially observed solutions of SPDE}
\label{ssec:SPDE}

Following \cite{Jentzen2009}, consider an unknown signal modelled by the solution of the SPDE 
\begin{subequations}
\label{eq:SPDE}
\eqnl{eq:SPDE_domain}{
\dfrac{\partial u}{\partial t} = \dfrac{\partial^2 u}{\partial x^2} + \dfrac{1}{2}u + \dfrac{\partial W_t}{\partial t}
}
defined on a domain $D = [0,1]$ and time $t \in [0,T]$, with the initial condition
\eqnl{eq:SPDE_boundary}{
u(x,0) = \sum_{n \geq 1} \dfrac{1}{n} e_n(x),
}
\end{subequations}
where the eigenfunction $e_n(x) = \sqrt{2} \sin(n \pi x)$ has corresponding eigenvalue $n^2\pi^2$, and with a cylindrical Brownian motion
\eqns{
W_t = \sum_{n \geq 1} \sqrt{q_n} e_n \beta^n_t,
}
where the terms $\beta^n_t$, $n \geq 1$, are independent Brownian motions. The final time $T$ is set to $T = 0.1$ and the variance of the Brownian motion $q_n$ is set to $q_n = 0.01$. The SPDE is observed at times $t_k = Tk/K$ for $k \in \{1,\dots,K\}$ and at the locations $o_l = l/(K'+1)$ for $l \in \{1,\dots,K'\}$ under an additive Gaussian noise with standard deviation $\sigma = 0.025$, for some integers $K$ and $K'$. The observation vector at time $t_k$ is denoted $y_k$ and is made of the scalar observations made at the locations $o_l$, $l \in \{1,\dots,K'\}$. The corresponding likelihood function is
\eqns{
g(y_k \given x_k) = \prod_{l=1}^{K'} \dfrac{1}{\sqrt{2\pi}\sigma} \exp\Big( -\dfrac{1}{2 \sigma^2} \big( y_{k,l} - x^l_k \big)^2  \Big),
}
where $y_{k,l}$ is the $l$\textsuperscript{th} component of the vector $y_k$ and where $x^l_k$ is the value of the SPDE at time $k$ and location $o_l$, for any $l \in \{1,\dots,K'\}$. An example of the set of observations obtained on one solution of \cref{eq:SPDE} is given in \cref{fig:obsSPDE}. To ease the estimation procedure, the standard deviation of the observation noise is taken $4$ times bigger in the UMIMC and  MIMC recursions.

\begin{figure}
\thispagestyle{empty}
\centering
\includegraphics[width=.65\textwidth,trim=50pt 260pt 70pt 270pt, clip]{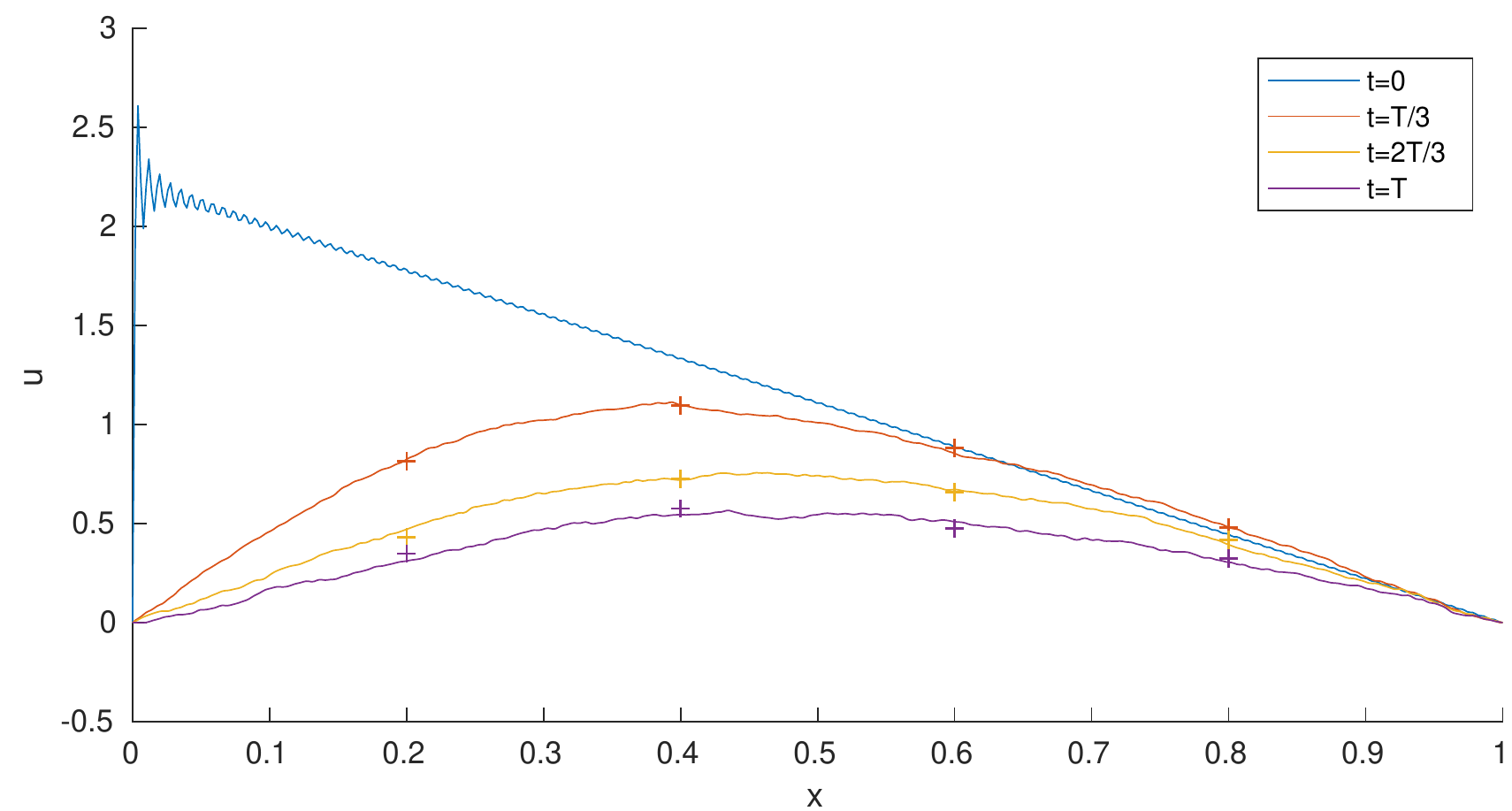}
\caption{Solution of \cref{eq:SPDE} for a given realisation of the Brownian motion, displayed at times $\{t_k\}_{k=1}^K$ with $K = 3$ and at locations $\{o_l\}_{l=1}^{K'}$ with $K' = 4$ and where the observations are indicated by crosses.}
\label{fig:obsSPDE}
\end{figure}

At index $\bsalpha \in \bbN^2_0$, this SPDE is solved using the exponential Euler scheme of \cite{Jentzen2009} with the first $2 \times 2^{\bsalpha_1}$ eigenfunctions and $2^{\bsalpha_2}$ time steps. The MIMC is used in the same way as in the previous section, but with a tolerance $\mathrm{TOL} = 5\times 10^{-3}$. The quantity to estimate is the integral of the true path at the last time step, so that
\eqns{
\bbE[S] = \dfrac{ \bbE\Big[ \varphi(X') \prod_{k = 1}^K g(y_k \given X'_k) \Big] }{ \bbE\Big[ \prod_{k = 1}^K g(y_k \given X'_k) \Big] },
}
where the expectation is w.r.t.\ to the path $X'$, where $\varphi(X')$ is the integral of the path $X'$ and where $X'_k$ is the vector containing the values of the path $X'$ at the locations $\{o_l\}_{l=1}^{K'}$ and at time $t_k$ for some $k \in \{1,\dots,K\}$. The experiments are run with $K = 3$ and $K' = 4$. The results displayed in \cref{fig:TOLvsWork2} show the same type of behaviour as in the simulation study of \cref{ssec:PDE} in spite of the differences between the two considered problems, e.g.\ in the nature of the approximation represented by the indices (spatial discretization in 2 dimensions vs.\ time discretization plus number of basis functions).

\begin{figure}
\thispagestyle{empty}
\centering
\begin{subfigure}[b]{0.72\textwidth}
\includegraphics[width=\textwidth,trim = 105pt 290pt 115pt 300pt, clip]{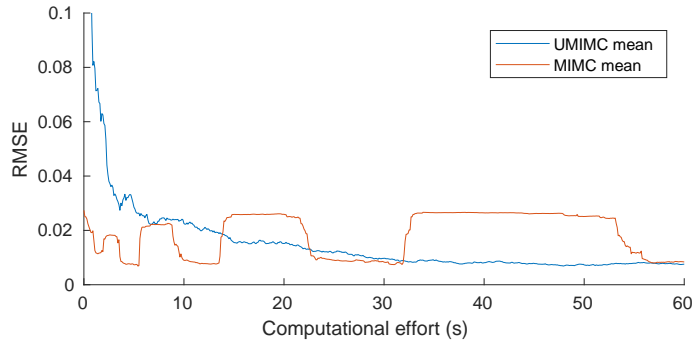}
\caption{Maximum level $(5,5)$}
\label{fig:TOLvsWork2_3}
\end{subfigure}
\begin{subfigure}[b]{0.72\textwidth}
\includegraphics[width=\textwidth,trim = 105pt 290pt 115pt 300pt, clip]{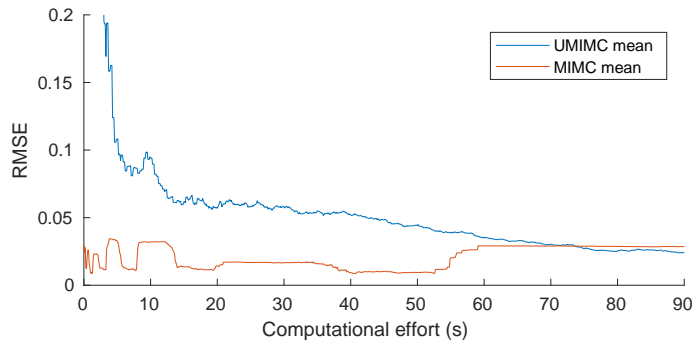}
\caption{Maximum level $(7,7)$}
\label{fig:TOLvsWork2_4}
\end{subfigure}
\caption{RMSE of the MIMC and UMIMC algorithms as a function of the computational effort for the problem~\cref{eq:SPDE}. The results are averaged over $50$ Monte Carlo runs.}
\label{fig:TOLvsWork2}
\end{figure}

\section{Summary}\label{sec:summary}

In this article, we have considered exact approximation of expectations associated to probability laws with discretizations in multiple dimensions. We have developed several optimality results and implemented the methodology to a couple of numerical examples. 

Future work associated to this methodology, includes combining our method in scenarios for which independent sampling from the (discretized) multi-index target is not possible. For instance, where one has to use Markov chain or sequential Monte Carlo methods (e.g.~\cite{besk:15} in the case of a single index). The analysis in such a scenario is of interest as is its application, to enhance the range of examples where our approach can be implemented. This is being conducted in \cite{jasra}.

\subsubsection*{Acknowledgements}

JH \& AJ were supported by an AcRF tier 2 grant: R-155-000-143-112. AJ is affiliated with the CQF, RMI and OR cluster, NUS. DC was partially supported by the EPSRC grant: EP/N023781/1.

\end{document}